\documentclass[10pt]{article} 
\usepackage{simpleConference}
\usepackage{fullpage, color, graphicx, enumitem, mathtools, tikz, thmtools, float, centernot}
\usepackage{times}
\usepackage{graphicx}
\usepackage{amssymb}
\usepackage{url,hyperref}
\usepackage[utf8]{inputenc}
\usepackage{csquotes}
\usepackage{amsmath}
\usepackage{amsfonts}
\usepackage{amsthm}
\usepackage{mathrsfs}
\usepackage{adjustbox}
\usepackage{enumitem}
\usepackage{xcolor}
\usepackage{float}
\usepackage{biblatex}
\usepackage{caption}
\usepackage{booktabs}
\usepackage{setspace}
\usepackage{subfigure}

\usepackage{algorithm}
\usepackage{algpseudocode}

\usepackage{color,soul}

\usepackage{float}
\usepackage{breqn}  % Used for breaking equations on multiple lines
\usepackage{booktabs} % Top and bottom rules for tables
\usepackage{enumitem} % Used to reduce itemize/enumerate spacing
\usepackage{palatino} % Use the Palatino font
\def \[#1\] {
  \begin{align*}
  #1
  \end{align*}
}
\newcommand{\eq}[1]{\begin{align*}#1\end{align*}}
\newtheorem{theorem}{Theorem}[section]

\newtheorem{observation}{Observation}[section]
\newtheorem{prop}{Proposition}[section]
\newtheorem{algo}{Algorithm}[section]
\newtheorem{lemma}{Lemma}[section]
\theoremstyle{definition}
\newtheorem{definition}{Definition}[section]
\newtheorem{remark}{Remark}[section]
\makeatletter  
\def\@endtheorem{\qed\endtrivlist\@endpefalse } % insert `\qed` macro
\makeatother
\newtheoremstyle{proof}
  {0pt}
  {\topsep}% measure of space to leave below the theorem. E.g.: 3pt
  {}% name of font to use in the body of the theorem
  {0pt}% measure of space to indent
  {\itshape}% name of head font
  {.}% punctuation between head and body
  { }% space after theorem head; " " = normal interword space
  {\thmname{#1}\thmnumber{ #2}\textnormal{\thmnote{ (#3)}}}
\theoremstyle{proof}
\newtheorem*{prf}{Proof}
\newcommand{\lem}[2]{
\begin{lemma} #1
\end{lemma}
\begin{prf} #2
\end{prf}
}
\newcommand{\prp}[2]{
\begin{prop} #1
\end{prop}
\begin{prf} #2
\end{prf}
}

\newcommand{\R}{\mathbb{R}}

\renewcommand{\P}{\mathbf{P}}

\newcommand{\E}{\mathbf{E}}

\newcommand{\sbe}{\subseteq}

\newcommand{\ce}{\coloneqq}
\newcommand{\set}[1]{\{#1\}}

\newcommand{\prn}[1]{\left( #1 \right)}

\newcommand{\ep}{\epsilon}

\newcommand{\bb}{\mathbf{b}}
\newcommand{\bn}{\mathbf{b}_{-i}}
\newcommand{\vv}{\mathbf{v}}

\newcommand{\pp}{\mathbf{p}}

\DeclareMathOperator*{\argmax}{arg\!\max}

\begin{document}

\parindent=0mm
\parskip=2mm
\setlist[itemize]{noitemsep, topsep=0pt}

\title{Transaction Fee Mining and Mechanism Design}

\author{\textbf{Michael Tang \textsuperscript{1} \quad\quad Alex Zhang\textsuperscript{1}}
\\
\textsuperscript{1}Department of Computer Science, Princeton University
\\
\small
\texttt{\{mwtang, alzhang\}@princeton.edu}
}

\maketitle

\thispagestyle{empty}

\begin{abstract}
Transaction fees represent a major incentive in many blockchain systems as a way to incentivize processing transactions. Unfortunately, they also introduce an enormous amount of incentive asymmetry compared to alternatives like fixed block rewards. We analyze some of the incentive compatibility issues that arise from transaction fees, which relate to the bids that users submit, the allocation rules that miners use to choose which transactions to include, and where they choose to mine in the context of longest-chain consensus. We start by surveying a variety of mining attacks including undercutting, fee sniping, and fee-optimized selfish mining. Then, we move to analyzing mechanistic notions of user incentive compatibility, myopic miner incentive compatibility, and off-chain-agreement-proofness, as well as why they are provably incompatible in their full form. Then, we discuss weaker notions of nearly and $\gamma$-weak incentive compatibility, and how all of these forms of incentive compatibility hold or fail in the trustless auctioneer setup of blockchains, examining classical mechanisms as well as more recent ones such as Ethereum's EIP-1559 mechanism and \cite{chung}'s burning second-price auction. Throughout, we generalize and interrelate existing notions, provide new unifying perspectives and intuitions on analysis, and discuss both specific and overarching open problems for future work.
\end{abstract}

\tableofcontents

\section{Introduction}
The design choices surrounding decentralized cryptocurrencies have recently become a point of great interest both due to a menagerie of emerging challenges and their far-reaching consequences on real blockchain systems with billions of USD worth of transaction volume and tens of millions of users. The primary design of modern cryptocurrencies relies on a distributed ledger, which replaces trusted third parties with hashes, cryptography, and mechanism design. First introduced in \cite{nakamoto2008bitcoin}, the basic idea is that participants in the network, known as miners, will try to solve a computationally-difficult hash puzzle\footnote{We motivate the problem using Proof of Work here, but almost everything we discuss also applies to alternate formulations such as Proof of Stake}. If a miner is successful, they ``win" the responsibility of adding the next block to the ledger, which includes adding the next set of transactions to the ledger. The miner is then rewarded with a fixed block reward and a set of transaction fees corresponding to the transactions they included in their block. In many real-world systems like Bitcoin, the long-term goal is to eventually phase out the block reward (for deflationary or other reasons), until eventually transaction fees become the primary incentive for miners to continue working on the network. 

Transactions in the blockchain system are validated when they are put on a newly mined block and acknowledged by other participants. Transaction fees are paid by users making a transaction as a way of incentivizing miners to use their computational resources to keep the blockchain running; however, unlike fixed block rewards, transaction fees are not fixed, and users can offer variable transaction fee amounts, e.g. based on how urgently they want their transaction to be serialized on the blockchain. The simplest and most common mechanism is a first-price auction system where, users bid according to their personal value of having their transaction being added to the next block, and miner as the untrusted auctioneer is incentivized to include higher-paying transactions into their blocks. Unfortunately, it tunrs out that this formulation trades off simple or honest user bids in favor of honest transaction selection on the part of the miner, and the optimal bidding strategy for users is less than obvious. Similarly, having mined blocks representing different values due to the variability of transaction fees also incentivizes various deviations from the honest mining strategy, which are effectively undetectable due to the latency and other extenuating idiosyncracies of the distributed system. Sadly, mechanistic and mining deviations both decrease the predictability, capacity, and robustness of the ledger as a whole, among other implications.

% As transaction fees become more prevalent, it is important to reason about their impact on the stability of the blockchain protocol, as well as potential design choices in the transaction-fee auction system.
In this paper, we survey major attacks under transaction fees that relate to \textit{mining}, such as undercutting \cite{undercut, 10.1007} and fee sniping \cite{Todd2013}, and point out possible defenses and holes in their analysis that can be further investigated in future work. We also explore \textit{mechanism design} questions surrounding incentive compatibility in the transaction-fee auction, proving a strong result by \cite{chung} regarding the impossibility of retaining classical incentive compatibility for all parties, and discussing the implications and alternative objectives to optimize in a not-fully-incentive-compatible world. We concretely survey several classical mechanisms \ref{classical} as well as state-of-the-art mechanisms \ref{eip} \ref{burning} and analyze their behavior and that of potential variants. We conclude by giving some final intuitions on the design space and open questions for future work.

% \hl{TODOS}
% SOME NOTES:

% on the mining side:
% - gave major attacks
% - explored defenses
% - synthesized new ideas future work

% on the transaction fee mechanism side:
% - motivated and defined incentive compatibility
% - explored relationships and implications of these notions of incentive compatibility
% - under intractability, explored weaker and stronger notions and tradeoffs and new generalizations
% - gave classical mechanisms and analysis
% - gave SOTA (EIP-1559, burning second-price) mechanisms and analysis
% - explored variants and new ideas
% - synthesized future work

\section{The Mining Game}
\subsection{Transaction-fee mining setup}
We first formalize a general model setup to represent the mining game with a focus on transaction fees as the block reward. We assume a set of $n$ miners $m_1,...,m_n$, each of which has some proportional mining power $\chi(m_i)$ such that $\sum_{i=1}^n \chi(m_i) = 1$. Each miners $m_i$ is aware of a directed tree $G(m_i)$ that represents their (private) knowledge of the blockchain, and they can choose any node/block on this tree to mine on. After a time interval that is exponentially distributed with mean $\chi(m_i)^{-1}$, the miner will find a new block $B$ and add a directed edge to the node they were mining on. They may then choose to broadcast this information to all other miners at any point. Our setup will be time-driven, so we observe and analyze the state of the system at some time $t$.

% [to Alex: rewrite + redefine these notes as desired!]
% basic setup: every miner has mining power $x(m)$, each miner is aware of directed tree $G(m)$, total of $t$ transaction fees arrive in $[0,t]$
% find a new block with Exponential distribution, create a new node and a new directed edge from that node to any existing node in $G(m)$
% include fees $F(B)$

\subsection{Mining notation} 
We will analyze the decision-making processes of a set of miners at some time $t$ during the game. Each miner will have access to a directed tree $G$ representing the publicly known chain, and a directed tree $G(m_i)$ representing their private chain. Let $T$ denote the total set of transaction fees. We let $\mathcal{T}(B)$ denote the transaction fees included in block $B$, and $\mathcal{R}(B)$ denote the announced transactions fees not included in $B$ or any of its predecessors (i.e. the remaining transaction fees after block $B$). We also use $h(B)$ to denote the height of a block $B$, or the length of the chain that ends at $B$, and $\mathcal{H}$ denote the height of the highest announced block. For private chains, we use $\mathcal{H}_{m_i}$ to denote the highest block on $m_i$'s private chain. Because general rational strategies will not choose to mine two blocks at the same height, we let $b_{m_i}^j$ denote the unique block that miner $m_i$ mined at height $j$ if it exists or the block at height $j$ that the miner first heard about.

\subsection{Honest Miner algorithm}
All future decisions and strategies will be time-driven, in the sense that any algorithm will make decisions at every time-step about what blocks to mine on, what transaction fees to include, and whether or not to publicly announce these blocks. The canonical "honest miner" algorithm is the ideal behavior for miners on the chain
\begin{algo} \label{honest}
The \texttt{HonestMiner} strategy for miner $m_i$ will
\begin{enumerate}
    \item Mine on the end of the longest chain. If there are multiple, choose the one $m_i$ heard about first, which is $b_{m_i}^{\mathcal{H}}$.
    \item Take all remaining available transaction fees $\mathcal{R}(b_{m_i}^{\mathcal{H}})$ upon discovering a new block $B$.
    \item Publish the new block to the public chain immediately.
\end{enumerate}
\end{algo}
and can be formalized into the pseudocode below in Algorithm 1. \\
\begin{algorithm} 
\caption{Honest Miner}
\begin{algorithmic}[1] 
\Procedure{honestmine}{}       
    \State $G \leftarrow$ publicly known blocks
    \State $G(m_i) \leftarrow$ publicly known blocks
    \State $\mathcal{H} \leftarrow$ height of highest publicly known block
    \State mine at $h^{\mathcal{H}}_{m_i}$
    \While{indefinitely}  \Comment{primary loop}
    \If{found a block $B$}
        \State Take transaction fees $\mathcal{R}(h^{\mathcal{H}}_{m_i})$
        \State add $B$ to $h^{\mathcal{H}}_{m_i}$ in both $G$ and $G(m_i)$ and announce to the public chain
        \State set $h^{\mathcal{H}+1}_{m_i} = B$, iterate $\mathcal{H}$ by $1$
        \State  mine at $h^{\mathcal{H}}_{m_i}$
    \ElsIf{other miner found block $B$}
        \State add $B$ to $h^{\mathcal{H}}_{m_i}$ in both $G$ and $G(m_i)$ and announce to the public chain
        \State set $h^{\mathcal{H}+1}_{m_i} = B$, iterate $\mathcal{H}$ by $1$
        \State  mine at $h^{\mathcal{H}}_{m_i}$
    \EndIf
    \EndWhile  \label{honest loop}
\EndProcedure
\end{algorithmic}
\end{algorithm}
\\
All other algorithms discussed in \textbf{Section 3} can be fit into the pseudocode framework described in Algorithm 1 with minor modifications, so we will only write out the algorithm details for them.

\begin{observation} \label{mininggap}
Under the transaction-fee regime, the \texttt{HonestMiner} strategy is not profitable for a time interval of length $t$ after some block is mined.
\end{observation}
Dubbed the \textit{mining gap} in \cite{undercut}, the point of Observation \ref{mininggap} is that for some time after a block is found, there is a non-trivial mining cost that outweighs the profit from transaction fees for discovering a new block. This observation may incentivize miners to use other non-honest strategies that make the blockchain more vulnerable to the attacks we will mention in the next section, and further motivate the design of new mechanisms in \textbf{Section 4}. 

\section{Transaction Fee Mining Attacks} \label{attacks}
\subsection{Undercutting}
First proposed in \cite{undercut}, the undercutting attack is performed when a miner forks the head of a chain and leaves a subset of the transaction fees unclaimed to encourage other miners to mine on top of the fork. They also introduce an simple alternative to the canonical "honest mining" strategy, which they call the \texttt{PettyCompliant} strategy. The idea is to follow the same strategy as the "honest mining" strategy, except if there is a tie in the longest chain, the \texttt{PettyCompliant} strategy will choose to mine on the block with the most available transaction fees rather than the oldest. More formally,
\begin{algo} \label{pettycompliant}
The \texttt{PettyCompliant} strategy for miner $m_i$ will
\begin{enumerate}
\item Mine on $b^{\mathcal{H}}_{*} = \argmax_{b^{\mathcal{H}}_{m_i}} \mathcal{R}(b_{m_i}^{\mathcal{H}})$, which is the block at height $\mathcal{H}$ with the highest available transaction fees.
\item Take all remaining available transaction fees $\mathcal{R}(b^{\mathcal{H}}_{*})$ upon discovering a new block $B$.
\item Publish the newly discovered block to the public chain immediately.
\end{enumerate}
\end{algo}
It is easy to observe that during the existence of a fork, $\texttt{PettyCompliant}$ performs strictly better than the honest miner because it mines on a block with higher rewards without the risk of losing it. Undercutting strategies can exploit this fact by incentivizing \texttt{PettyCompliant} miners to extend their forked blocks, especially in the event of a tie for the highest block. A simple undercutting strategy proposed in \cite{undercut} is called \texttt{LazyFork} which will fork the head of the chain if it is more valuable to take the transaction fees inside. More formally,
\begin{algo} \label{lazyfork}
The \texttt{LazyFork} strategy for miner $m_i$ will
\begin{enumerate}
\item Consider $b^{\mathcal{H}}_{*} = \argmax_{b_{m_i}^{\mathcal{H}}} \mathcal{R}(b_{m_i}^{\mathcal{H}})$, which is the block at height $\mathcal{H}$ with the highest available transaction fees and consider $b^{\mathcal{H}-1}_{*} = \argmax_{b_{m_i}^{\mathcal{H}-1}} \mathcal{R}(b_{m_i}^{\mathcal{H}-1})$. Define $\Delta_{\mathcal{H}} = \mathcal{R}(b^{\mathcal{H}-1}_{*}) - \mathcal{R}(b^{\mathcal{H}}_{*})$, i.e. the maximum transaction fees you could get from forking on top of $b^{\mathcal{H}-1}_{*}$.
\begin{enumerate}
\item If $m_i$ owns $b^{\mathcal{H}}_{*}$, $m_i$ will mine on top of $b^{\mathcal{H}}_{*}$. 
\item Else if $\mathcal{R}(b^{\mathcal{H}}_{*}) \geq \Delta_{\mathcal{H}}$, $m_i$ will mine on top of $b^{\mathcal{H}}_{*}$ because it is not profitable enough to fork.
\item Else $m_i$ will mine on $b^{\mathcal{H}-1}_{*}$ because it is profitable enough to fork.
\end{enumerate}
\item Take $\frac{1}{2}$ of the remaining available transaction fees upon discovering a new block $B$.
\item Publish the newly discovered block to the public chain immediately.
\end{enumerate}
\end{algo}
The existence of $\texttt{LazyFork}$ miners is dangerous because it increases the probability of orphaned blocks existing, as well as the vulnerability of the blockchain to a double-spending attack. However, the $\texttt{LazyFork}$ strategy itself is susceptible to more aggressive undercutting by other miners, and therefore a more aggressive modification called $\texttt{FunctionFork}$ uses a function $f(\cdot)$ to determine the amount of transaction fees to leave open to other miners, and whether or not they should continue honestly mining on top of the head of $G$, or perform an undercutting attack. 
\begin{remark}
While another option for an undercutting strategy is to choose to mine on top of blocks lower than $\mathcal{H} - 1$, for ease of analysis the proposed strategies only consider mining on $\mathcal{H} - 1$ or $\mathcal{H}$. While mining at a lower level on the chain is riskier, there is potentially interesting equilibrium behavior that may arise from opening up a wider set of options for an undercutting miner.
\end{remark}

\begin{algo} \label{functionfork}
The \texttt{FunctionFork} strategy for miner $m_i$ will
\begin{enumerate}
\item Define the following variables
\begin{itemize}
\item $b^{\mathcal{H}}_{*} = \argmax_{b_{m_i}^{\mathcal{H}}} \mathcal{R}(b_{m_i}^{\mathcal{H}})$, which is the block at height $\mathcal{H}$ with the highest available transaction fees.
\item $b^{\mathcal{H}-1}_{*} = \argmax_{b_{m_i}^{\mathcal{H}-1}} \mathcal{R}(b_{m_i}^{\mathcal{H}-1})$, which is the block at height $\mathcal{H}-1$ with the highest available transaction fees. 
\item $\Delta_{\mathcal{H}} = \mathcal{R}(b^{\mathcal{H}-1}_{*}) - \mathcal{R}(b^{\mathcal{H}}_{*})$, i.e. the maximum transaction fees $m_i$ could get from forking on top of $b^{\mathcal{H}-1}_{*}$.
\item $\mathcal{V}_{\text{CONT}}(f) = f(\mathcal{R}(b^{\mathcal{H}}_{*}))$ which is the value for continuing on the highest chain.
\item $\mathcal{V}_{\text{UND}}(f) = \min\{f(\mathcal{R}(b^{\mathcal{H}-1}_{*})), \Delta_{\mathcal{H}}\}$ which is the value for undercutting.
\end{itemize}
\item Choose where to mine based on
\begin{enumerate}
\item If you own $b^{\mathcal{H}}_{*}$, mine on top of $b^{\mathcal{H}}_{*}$. 
\item Else if $\mathcal{V}_{\text{CONT}}(f) \geq \mathcal{V}_{\text{UND}}(f)$, mine on top of $b^{\mathcal{H}}_{*}$ because it is not profitable enough to fork.
\item Else mine on $b^{\mathcal{H}-1}_{*}$ because it is profitable enough to fork.
\end{enumerate}
\item If you mined on top of $b^{\mathcal{H}}_{*}$, take $\mathcal{V}_{\text{CONT}}(f)$ transaction fees, and $\mathcal{V}_{\text{UND}}(f)$ otherwise.
\item Publish the newly discovered block to the public chain immediately.
\end{enumerate}
\end{algo}
Clearly, $f(\cdot)$ is supposed to be some mapping from the transaction fees you receive to the amount you give, and therefore any general strategy will make it a monotonically increasing function. \cite{undercut} show that when considering the linear family of functions $f(x) = kx$ for $k \in [0,1]$, if we assume that every miner is \textbf{non-atomic}, i.e. there are infinite miners with infinitesimally small hashing power, so they only locally consider how to maximize the block they just found because with zero probability they find another, then the equilibrium strategy is to undercut other players by slowly decreasing $k$ until it hits zero. They also show that under the assumption that every miner is \textbf{atomic}, the equilibrium behavior is unstable. They further prove that there exists a family of functions $\mathcal{F}$ in which it is an equilibrium for every miner to use \texttt{FunctionFork}$(f)$ for $f \in \mathcal{F}$ under strong conditions, which we state without proof:
\begin{theorem} \label{equilibrium}
Assume every miner is non-atomic, and that miners may only mine on top of $\mathcal{H}$ and $\mathcal{H}-1$. Let $W_0$ be the upper branch of the Lambert function $W$ satisfying $W_0(xe^x) = x$ for all $x \in [-\frac{1}{e}, \infty)$ and $W_0(x) \in [-1,\infty)$. Then, for any constant $\gamma \leq \frac{1}{2}$ such that $2 \gamma - \ln(\gamma) \geq 2$, define the monotonically-increasing function $f(x)$ as
\[ f(x) = \begin{cases} 
          f(x) = x & x \leq \gamma \\
          f(x) = -W_0(-\gamma e^{x - 2 \gamma}) & \gamma < x < 2 \gamma - \ln(\gamma) - 1 \\
          f(x) = 1 & x \geq 2 \gamma - \ln(\gamma) - 1
       \end{cases}
 \]
Then it is an equilibrium for every miner to use the \texttt{FunctionFork}(f) strategy, i.e. it is a best response over all potential other strategies. 
\end{theorem}

Finally, \cite{undercut} conclude using Theorem \ref{equilibrium} and a set of simulations that while it is intractible to prove that convergence to such equilibrium is guaranteed in practice, they are able to simulate scenarios where rational miners eventually change to undercutting eachother as the equilibrium, severely damaging the stability of the blockchain.

\begin{remark} An important caveat to the analysis above is that it ignores block size limits, which are the maximum claimable fees on any block. \cite{10.1007} finds that the undercutting scheme is significantly weakened than exactly what was proposed above under these more realistic conditions because it makes undercutting less profitable in expectation. This fact also somewhat implies that any defenses against undercutting should restrict the size of transaction fees that can be found on a block, or restrict what transaction fees can be put on a forked block in the first place. \cite{Todd2013} was added to Bitcoin for this purpose, which ensures transactions cannot be put into a block unless it is of a certain height on the chain, preventing an attacker from putting all high-paying fees onto a single forked block.
\end{remark}

\subsection{Fee sniping}
Fee sniping (as described in \cite{Todd2013, 10.1007/978-3-319-70278-0_17}) is an attack where a miner will deliberately fork a block to take the high transaction fees found on the currently accepted block. Although it is unlikely that any miner will re-mine a previous block and find a new block before a new block is found by another miner, because transaction fees are not fixed, it is sometimes profitable in expectation to attempt a fee snipe. We define a simple fee sniping strategy below
\begin{algo} \label{feesnipe}
The \texttt{BasicFeeSnipe} strategy for miner $m_i$ will
\begin{enumerate}
\item Consider $b^{\mathcal{H}}_{*} = \argmax_{b_{m_i}^{\mathcal{H}}} \mathcal{T}(b_{m_i}^{\mathcal{H}})$, which is the block at height $\mathcal{H}$ with the highest available transaction fees and consider $b^{\mathcal{H}-1}_{*} = \argmax_{b_{m_i}^{\mathcal{H}-1}} \mathcal{T}(b_{m_i}^{\mathcal{H}-1})$. Define $\Delta_{\mathcal{H}} = \mathcal{R}(b^{\mathcal{H}-1}_{*}) - \mathcal{R}(b^{\mathcal{H}}_{*})$, i.e. the maximum transaction fees you could get from forking on top of $b^{\mathcal{H}-1}_{*}$.
\begin{enumerate}
\item If $m_i$ owns $b^{\mathcal{H}}_{*}$, $m_i$ will mine on top of $b^{\mathcal{H}}_{*}$. 
\item Else if $\mathcal{R}(b^{\mathcal{H}}_{*}) \geq \chi(m_i)^2 \cdot \mathcal{R}(b^{\mathcal{H}-1}_{*})$, $m_i$ will mine on top of $b^{\mathcal{H}}_{*}$ because it gets a higher expected reward by being honest.
\item Else $m_i$ will mine on $b^{\mathcal{H}-1}_{*}$ because it gets a higher expected reward by forking.
\end{enumerate}
\item Take all of the remaining available transaction fees upon discovering a new block $B$.
\item Publish the newly discovered block to the public chain immediately.
\end{enumerate}
\end{algo}
We can observe that \texttt{BasicFeeSnipe} is almost the same as \texttt{LazyFork} discussed in \textbf{Section 3.1}, except it decides where to mine based on the expected value of winning the next two blocks by forking and not on incentivizing other miners. We see that in the forking condition $\chi(m_i)^2 \cdot \mathcal{R}(b^{\mathcal{H}-1}_{*}) > \mathcal{R}(b^{\mathcal{H}}_{*})$, where the left side of the inequality is precisely the probability of winning the next two blocks $\chi(m_i)^2$ multiplied by the current transaction-fee rewards for winning $\mathcal{R}(b^{\mathcal{H}-1}_{*})$ at the current timestep, which is actually a safe estimate of the expectation because during the mining process, more fees will come in.

\begin{remark}
To combat fee sniping, \cite{Todd2013} added an nLockTime variable on the Bitcoin blockchain to forcibly prevent miners from orphaning the current best block by forking an earlier block. They take advantage of block size limits to ensure one cannot push an absurdly large amount of transaction fees to a single block.
\end{remark}

\subsection{Whale transaction attacks} \cite{10.1007/978-3-319-70278-0_17} analyzes a type of attack performed in which a miner $A$ makes a valid transaction with another miner $B$ on the main blockchain, then forks the chain to a point before the transaction was valid to nullify it. To incentivize other miners to work on this new chain, they offer an exorbitant transaction fee that can be claimed if the fork becomes the new main chain. They call this attack a \textit{whale transaction}, and demonstrate that it also can lead to vulnerabilities in the blockchain like double spending. 
\\\\
Consider the following simplified setup for miner $A$ attempting to perform a whale transaction:
\begin{enumerate}
\item \textbf{Assumptions:} Let $1$ be the normalized average block reward (fixed block reward + average transaction fee) and let $\delta$ be the (normalized) whale fee that miner $A$ offers. Assume hashing power $\chi(x_i)$ is fixed for all miners, and relative mining power is known among all miners. Also assume miners will choose the most profitable strategies, and that they can either be honest or work on miner $A$'s fork. Finally, to allow for easier steady-state analysis, assume at any point that the relative mining power of miner $A$'s fork is fixed, i.e. they will offer more of a fee $\delta$ to keep the mining power constant. 
\item \textbf{Pre-mining phase:} Miner $A$ will make a transaction with miner $B$, then attempt to undo it by forking at an earlier node in $G$. Miner $A$ can then issue several illegimate transactions such as a double spending transaction in $G(A)$ until they are ready to announce their blocks. They will also ensure that they have hefty unclaimed transaction fees on their fork.
\item \textbf{Racing phase:} Assume $\mathcal{H}_A < \mathcal{H}$, so miner $A$'s fork has to catch up to the main blockchain. Let $\gamma_i = 1$ if miner $m_i$ decides to mine on miner $A$'s fork, and $0$ otherwise. Then the relative mining power on miner $A$'s fork is $q = \chi(A) + \sum \gamma_i \chi(m_i)$ and the relative mining power on the honest chain is $p = 1 - q$. Let $z$ denote the lead that the main branch has over $A$'s fork. Then, the probability that the fork overtakes the main branch, i.e. $z = -1$, can be computed by modeling this process as a biased random walk where $z$ decrements with probability $q$ and increments with probability $p$. The probability that the fork eventually overtakes the main branch is therefore
$$ a_z = \min\{1, q/p\}^{z+1} $$
\end{enumerate}
\begin{theorem} \label{whaleattack}
Suppose miner $A$ has mining power $\chi(A)$, and a miner $X$ has mining power $\chi(X)$. Assume every other miner is honest. For notational sake, let $M = \sum_{m \neq A} \chi(m)$ be the total mining power of all miners except $A$. Then under the setup defined above, miner $X$ is incentivized to mine on $A$'s fork when
$$ \delta > \frac{(1 - (\frac{\chi(A)}{M})^{z+1})}{M} \cdot \frac{\chi(A) + \chi(X)}{(\frac{\chi(A) + \chi(X)}{M + \chi(X)})^{z+1}} - 1 $$
\end{theorem}
\textbf{Proof:} 
Assume miner $X$ decides to stay honest. Then their expected reward is the probability that the fork fails, $(1 - a_z)$ multiplied by the reward (which is normalized to $1$) for finding a new block, and the probability of finding a new block, $\frac{\chi(X)}{M}$. So
$$ \E[X \text{ mines honest reward}] = (1 - a_z) \frac{\chi(X)}{M} = \frac{(1 - (\frac{\chi(A)}{M})^{z+1}) \cdot \chi(X)}{M} $$
Now assume $X$ decides to work on $A$'s fork. Their expected reward is the probability that the fork succeeds, $a_z$, multiplied by the reward $(1 + \delta)$ for finding a new block, and the probability of finding a new block, $\frac{\chi(X)}{q}$, where $q = \chi(A) + \chi(X)$ now. So,
$$ \E[X \text{ mines whale reward}] = a_z \frac{\chi(X)}{q} (1 + \delta) = \frac{(\frac{\chi(A) + \chi(X)}{M - \chi(X)})^{z+1}) \cdot \chi(X)}{\chi(A) + \chi(X)} \cdot (1 + \delta) $$
Now miner $X$ will be incentivized to mine on miner $A$'s fork if their expected reward is higher. This is precisely when
\[ \frac{(\frac{\chi(A) + \chi(X)}{M - \chi(X)})^{z+1}) \cdot \chi(X)}{\chi(A) + \chi(X)} \cdot (1 + \delta) > \frac{(1 - (\frac{\chi(A)}{M})^{z+1}) \cdot \chi(X)}{M} \\
\delta > \frac{(1 - (\frac{\chi(A)}{M})^{z+1})}{M} \cdot \frac{\chi(A) + \chi(X)}{(\frac{\chi(A) + \chi(X)}{M + \chi(X)})^{z+1}} - 1 
\\ \square \]
\cite{10.1007/978-3-319-70278-0_17} use Theorem \ref{whaleattack} and plug in values for $z$ and $\delta$ to find what relative values of $\delta$ are necessary to motivate a whale attack. It should be remarked that the values of $\delta$ are quite large, and therefore, a whale attack is only rational if miner $A$ is able to get a huge sum of rewards. However, as block rewards continue to diminish, it is also easily observable that the whale transaction can be smaller to promote the same effect.

\begin{remark}
The analysis above makes strong assumptions about the honest behavior of other miners on the chain, and therefore it is entirely possible that lower values of $\delta$ will allow for a whale transaction attack to be profitable for any deviant miners. Additionally, under the transaction-fee regime, as we will see in the next section, miners can use blocks with low-value fees as a starting point to attempt low-risk attacks. 
\end{remark}

\subsection{Fee-optimized selfish mining}
\cite{DBLP:journals/corr/EyalS13} first proposed the selfish mining algorithm, which can yield a higher expected reward than the honest miner even with $<51\%$ of the total hashing power. The intuitive idea is to waste other miners' hashing power by hiding the existence of their owned mined blocks, and only revealing them when the public chain reveals a new block. A transaction-fee regime version of this algorithm is defined below:
\begin{algo}
The \texttt{SelfishMiner} strategy for a miner $m_i$ will
\begin{enumerate}
\item Choose to mine on $\mathcal{H}_{m_i}$, which is the highest block on their private chain.
\item Take all remaining available transaction fees $\mathcal{R}(\mathcal{H}_{m_i})$ upon discovering a new block $B$.
\item If $h(B) = \mathcal{H}$, publish the block immediately. Otherwise, if there exists two distinct blocks of height $\mathcal{H}$ where one is owned by $m_i$, and $\mathcal{H}_{m_i} = \mathcal{H} + 1$, reveal $b_{m_i}^{\mathcal{H}}$.
\end{enumerate}
\end{algo}
In \cite{undercut}, they observe that under the transaction-fee regime, the default selfish mining algorithm will have a slightly higher expected reward than in the standard fixed-block reward regime despite performing the same actions.
\begin{theorem} \label{selfishminer}
Let $\gamma \in [0,1]$ be the probability that if the selfish miner $m_i$ is in a race with the public chain, it is not orphaned. Then given $\chi(m_i) = \alpha \in (0,0.5)$, the expected reward of \texttt{SelfishMiner} is
$$ \E(\text{reward}) = \frac{5\alpha^2 - 12\alpha^3 + 9\alpha^4 - 2\alpha^5 + \gamma(\alpha - 4 \alpha^2 + 6\alpha^3 - 5\alpha^4 + 2\alpha^5)}{2\alpha^3 - 4\alpha^2 + 1} $$
\end{theorem}
\textbf{Proof}: Consider the following states, in which at any timestep in the game, a miner $m_i$ running \texttt{SelfishMiner} will be in:
\begin{itemize}
\item State $0$: The highest block in $G$ is also the highest block in $G(m_i)$, i.e. miner $m_i$ is not hiding any blocks.
\item State $j > 0$: Miner $m_i$ has a private chain of length $j$, so $\mathcal{H}_{m_i} = \mathcal{H} + j$
\item $0'$: $\mathcal{H}_{m_i} = \mathcal{H}$, but one is produced by $m_i$, and the other is not, i.e. the selfish miner is racing with the public chain.
\end{itemize}
Let $f_s$ be the probability that a transaction is in the block found by a $m_i$ conditioned on being in state $s$, and $p_s$ is the probability that miner $m_i$ is in state $s$, then the expected reward for the miner is $\sum_s f_s \cdot p_s$. Observe that the probability of being in any state $s$ is identical in both the fixed reward and transaction-fee regime, and \cite{DBLP:journals/corr/EyalS13} have already shown that
\begin{align*}
& p_0 = \frac{1 - 2\alpha}{2 \alpha^3 - 4\alpha^2 + 1} \\
& p_0' = \frac{(1 - \alpha)(\alpha - 2\alpha^2)}{2 \alpha^3 - 4\alpha^2 + 1} \\
& p_j = \biggr(\frac{\alpha}{1 - \alpha}\biggr)^{j-1} \cdot \frac{\alpha(1 - 2\alpha)}{2 \alpha^3 - 4\alpha^2 + 1} \quad \text{ for } j > 0 \\
\end{align*}
\textbf{We now compute $f_s$.} Suppose a transaction arrives in state $0$. If the selfish miner $m_{i}$ finds the block, which occurs with probability $\alpha$, the transaction will be in the longest chain only if they mine the next block; this whole event occurs with probability $\alpha^2$. Otherwise, an honest miner $m_{-i}$ might find the next block with probability $(1-\alpha)$ and cause a race where whoever wins gets the transaction fee. The selfish miner wins this race with probability $\alpha + \gamma(1 - \alpha)$, so in total,
$$ f_0 = \alpha^2 + \alpha(1 - \alpha)(\alpha + \gamma(1 - \alpha))$$
For the $0'$ state, we observe that the selfish miner will get the transaction only if they find the next block first, which occurs with probability $\alpha$, so 
$$ f_{0'} = \alpha $$ 
The remaining states can be recursively computed: $f_1$ is precisely the probability of the selfish miner gaining another block to guarantee that the transaction will belong to $m_i$ plus the probability of returning to state $0'$ and winning. Formally,
$$ f_1 = \alpha + (1 - \alpha) f_{0'} = \alpha + (1 - \alpha) \alpha $$
To compute $f_j$ for $j > 1$, observe that the transaction will not belong to $m_i$ only if the honest miner wins the next $j - 1$ blocks, returning $m_i$ to state $0$, then winning the block in state $0$. This is precisely
$$ f_j = 1 - \biggr( (1 - \alpha)^{j-1} (1 - f_0) \biggr) \quad \text{ for } j > 1 $$
\textbf{Finally, we compute $\sum_s p_s f_s$} and condense some of the simple but tedious computations:
$$ \E(\text{reward}) = p_0f_0 + p_{0'}f_{0'} + p_1f_1 + \sum_{j > 1} p_j f_j $$
$p_0f_0 + p_{0'}f_{0'} + p_1f_1$ can be computed by multiplying the individual expressions and summing, so we will state their results without explicitly writing out the intermediate steps. The more interesting computation is
\begin{equation} \label{selfish1}
\sum_{j > 1} p_j f_j = \sum^\infty_{j = 2} \biggr[ \biggr(\frac{\alpha}{1 - \alpha}\biggr)^{j-1} \frac{\alpha - 2\alpha^2}{2\alpha^3 - 4\alpha^2 + 1} - \alpha^{j-1} (1 - f_0) \frac{\alpha - 2\alpha^2}{2\alpha^3 - 4\alpha^2 + 1}\biggr]
\end{equation}
By the properties of a geometric series, $\sum^\infty_{j = 2} (\frac{\alpha}{1 - \alpha})^{j-1} = \frac{\alpha}{1 - 2\alpha}$ and $\sum^\infty_{j = 2} \alpha^{j-1} = \frac{\alpha}{1 - \alpha}$, so we can simplify (\ref{selfish1}) to get
\[
\sum_{j > 1} p_j f_j = \biggr[ \biggr(\frac{\alpha}{1 - 2\alpha} \cdot \frac{\alpha - 2\alpha^2}{2\alpha^3 - 4\alpha^2 + 1} -  \frac{\alpha}{1 - \alpha} \cdot (1 - f_0) \frac{\alpha - 2\alpha^2}{2\alpha^3 - 4\alpha^2 + 1}\biggr] \\
= \frac{\alpha^2 - \alpha(1 + \alpha(1 - \alpha)(1 - \gamma))(\alpha - 2\alpha^3)}{2\alpha^3 - 4\alpha^2 + 1}
\]
and finally, we get our result by summing all the terms
\[
\E(\text{reward}) = \frac{2\alpha^2 - 5\alpha^3 + 2\alpha^4 + \alpha\gamma - 4\alpha^2 \gamma + 5\alpha^3\gamma - 2\alpha^4 \gamma}{2\alpha^3 - 4\alpha^2 + 1} + \frac{\alpha^2 - 3\alpha^3 + 2\alpha^4}{2\alpha^3 - 4\alpha^2 + 1} \\
+ \frac{2\alpha^2 - 5\alpha^3 + 2\alpha^4}{2\alpha^3 - 4\alpha^2 + 1} + \frac{\alpha^2 - \alpha(1 + \alpha(1 - \alpha)(1 - \gamma))(\alpha - 2\alpha^3)}{2\alpha^3 - 4\alpha^2 + 1} \\
= \frac{5\alpha^2 - 12\alpha^3 + 9\alpha^4 - 2\alpha^5 + \gamma(\alpha - 4 \alpha^2 + 6\alpha^3 - 5\alpha^4 + 2\alpha^5)}{2\alpha^3 - 4\alpha^2 + 1} \\ \square
\]
It is interesting to observe that for $\gamma \in [0,0.55]$, the expected reward found in Theorem \ref{selfishminer} is strictly greater than the expected reward under the block regime found in \cite{DBLP:journals/corr/EyalS13}, which is $\frac{\alpha(1 - \alpha)^2 (4 \alpha + \gamma(1 - 2\alpha)) - \alpha^3)}{1 - \alpha(1 + (2 - \alpha)\alpha)}$. \cite{undercut} reasons that this is because in high-number states, the selfish miner will gain a disproportionately large number of transaction fees while in the fixed-reward model, the rewards gained for mining a new block are always fixed. While interesting, the reward gain in the transaction-fee regime is minimal compared to the fixed-reward regime. A more interesting analysis comes from the fact that selfish miners can look at the value of a block when deciding whether or not to hide it. Consider the following fee-optimized selfish miner proposed in \cite{undercut}, which has a threshold parameter $\beta$:
\begin{algo}
The \texttt{FeeSelfishMiner} strategy for a miner $m_i$ will
\begin{enumerate}
\item Choose to mine on $\mathcal{H}_{m_i}$, which is the highest block on their private chain.
\item Take all remaining available transaction fees $\mathcal{R}(\mathcal{H}_{m_i})$ upon discovering a new block $B$.
\item If $h(B) = \mathcal{H}$ or if $\mathcal{T}(B) > \beta$, publish the block immediately. Otherwise, if there exists two distinct blocks of height $\mathcal{H}$ where one is owned by $m_i$, and $\mathcal{H}_{m_i} = \mathcal{H} + 1$, reveal $b_{m_i}^{\mathcal{H}}$.
\end{enumerate}
\end{algo}
The intuitive idea is as follows. If a selfish miner comes across a nearly worthless block with $\mathcal{T}(B) < \beta$, there is little risk in using it in a fork to potentially overtake the main chain. \cite{undercut} recompute both $p_s$ and $f_s$ and an additional state $0''$, which is where the \texttt{FeeSelfishMiner} will choose to honestly mine, and conclude with the following theorem, which we state without proof because the analysis is very similar to Theorem \ref{selfishminer}:
\begin{theorem} \label{feeselfishminer}
Let $\gamma \in [0,1]$ be the probability that if the selfish miner $m_i$ is in a race with the public chain, it is not orphaned. For a \texttt{FeeSelfishMiner} miner $m_i$ using threshold $\beta$, given $\chi(m_i) = \alpha \in (0,0.5)$, the expected reward is
$$ \E(\text{reward}) = \biggr(\frac{1 + \beta(1 - \alpha)^2 (1 - \gamma)}{e^{\beta} - 1} + 5\alpha + (1 - \alpha)^2 \gamma + \frac{2\alpha^2}{1 - 2\alpha} - 2\alpha^2 \biggr) \cdot \biggr(\frac{\alpha(1 - 2\alpha)(1 - e^{-\beta})}{1 - 2 e^{-\beta}\alpha - 3(1 - e^{-\beta})\alpha^2}\biggr) $$
\end{theorem}
Plugging in values into Theorem \ref{feeselfishminer}, it is observable that for $\alpha = \frac{1}{3}$, the \texttt{SelfishMiner} and \texttt{HonestMiner} schemes both achieve an expected reward of approximately $\frac{1}{3}$, but the \texttt{FeeSelfishMiner} with an optimal $\beta$ performs $13.6\%$ better with an expected reward of $\approx 0.38$. 

\begin{remark}
A key part of this analysis is the fact that blocks containing low-transaction fees are a useful way to start a forking attack because they are not very valuable to the attacker even if they are on the eventual chain. However, selfish mining itself offers no incentives for other miners to help, making it relatively weak unless a single miner owns a large percentage of the mining power. We believe that there is a lot of future work in analyzing attacks that take advantage of low transaction-fee blocks while also incentivizing other miners to mine on their fork to greatly increase the probability of a successful attack.
\end{remark}
 
\section{Transaction Fee Mechanism Design}
Besides preventing selfish-mining attacks, there are other desiderata to consider in the design of a transaction fee system. Here, we are specifically interested in the \textit{transaction fee mechanism} (TFM), and our goal is to modify or relax the first price auction setup in order to obtain or trade off against desired properties. To start, let us define the framework of a TFM.

\subsection{TFM definitions}
We consider a single auction instance corresponding to the next mined block:
\begin{enumerate}
\item $m$ users with values $\vv$ for having their transaction included in the block, submit bids \footnote{When discussing TFMs, we will use the terms ``bid" and ``transaction" interchangeably} $\bb$
\item the miner picks a subset of $k$ bids $B_I \sbe \bb$\footnote{For concision, throughout this paper we will abuse notation to treat vectors $\bb,\vv$ as sets} to include in the block
\item the blockchain then confirms a further subset $B_C \sbe B_I$ and enforces payments $\pp$ for each user whose transaction was confirmed, and revenue $r$ for the miner
\end{enumerate}
where the blockchain always executes the mechanism honestly, but the users and miner are strategic, acting to maximize their respective utility ($r$ for miner and $v_i - p_i$ for user $i$, with users with unconfirmed bids getting utility 0).

Note that all of this operates under the continuing assumption that the blockchain is effectively permissionless and anonymous.

\begin{remark}[Iterated analysis is hard] \label{iterated}
Ideally, we would prefer our analysis, including notions such as incentive compatibility and miner and user utility, to extend to an iterated game over multiple rounds of block proposals into account. Unfortunately, conditioning on large numbers of unpredictable future transactions makes such analysis becomes significantly more difficult, so the literature generally focuses on the mechanism applying to a single auction instance. This may be a rich area of further study, e.g. \cite{chung} points out that unconfirmed bids from fake bids added by miners are not easy to retract in blockchains like Ethereum and thus would have a chance of incurring the full cost of the bid in future rounds, which may serve as a penalty for encouraging honest miner behavior. We use this intuition in \ref{weakerIC}.
\end{remark}

%Following the notation in \cite{chung}, let $\mathbf{b} = (b_1,...,b_m)$ be the vector of bids.
Following the notation in \cite{chung}, which distinguishes between inclusion and confirmation, a TFM then corresponds to a tuple $(\mathbf{I}, \mathbf{C}, \mathbf{P}, \mathbf{M})$ where: ($\R_+ \ce [0, \infty) \sbe \R$)
\begin{enumerate}
    \item $\mathbf{I}: \R_+^m \to \R_+^m \times \set{0,1}^m$ is the miner's inclusion rule: given bids, output the bids and whether each bid is included
    \item $\mathbf{C}: \R_+^m \times \set{0,1}^m \to \set{0,1}^m$ is the blockchain's confirmation rule: given bids and an inclusion vector, output whether each bid is confirmed
    \item $\mathbf{P}: \R_+^m \times \set{0,1}^m \to \set{\R_+}^m$ is the blockchain's payment rule: given bids and a confirmation vector, output payments
    \item $\mathbf{M}: \R_+^m \times \set{0,1}^m \to \R_+$ is the blockchain's miner revenue rule: given bids and a confirmation vector, output miner revenue
\end{enumerate}
% We have $\mathbf{I}$ output the bid vector along with the inclusion vector so that it can be directly composed with downstream rules, which depend both on the original bid vector and a confirmation/inclusion vector.

Since the blockchain always implements its rules honestly, in analysis we can simplify this to a tuple $(\mathbf{x}, \mathbf{p}, \mu)$ where
\begin{enumerate}
    \item $\mathbf{x} \ce \mathbf{C} \circ \mathbf{I}: \R_+^m \to \set{0,1}^m$ is the allocation rule: given bids, output whether each bid is confirmed
    \item $\mathbf{p} \ce \mathbf{P} \circ \mathbf{I}: \R_+^m \to \R_+^m$ is the payment rule: given bids, output payments
    \item $\mu \ce \mathbf{M} \circ \mathbf{I}: \R_+^m \to \R_+$ is the miner revenue rule: given bids, output miner revenue
\end{enumerate}

When the mechanism is random, we relax whether a bid is included/confirmed into the \textit{probability} of inclusion/confirmation $[0,1]^m$, and output \textit{expected} payments and revenue.

% We say that a TFM results in nontrivial miner revenue if $\mu \ne \mathbf{0}$, i.e. there exist conditions under which the miner gets positive (expected) revenue.

Note that the strategy space includes deviations from honest behavior such as:
\begin{enumerate}
    \item Users can bid dishonestly, which is sometimes known as \textit{bid shading}, possibly after examining some or all other bids (note that as long as user payments are non-negative, i.e. the mechanism does not pay users for bidding, this implies underbidding)
    \item Users can submit additional fake bids (for which they have zero value), possibly after examining all real bids
    \item Miners can submit fake bids, possibly after examining all real bids
    \item Miners can select bids dishonestly with respect to the inclusion rule
    \item Cartels (i.e. sets) of users and the miner can collude to deviate in ways that maximize their joint utility (sum of individual utilities), including overbidding, underbidding, submitting additional fake bids, and selecting dishonestly
\end{enumerate}

\begin{remark}[Joint utility is all you need] A simple offchain payment suffices to ensure that whenever joint utility is strictly increased, payoff can be distributed among colluders such that each receives strictly greater utility.
\end{remark}

Due to idiosyncratic properties of blockchains (i.e. costless account creation, permissionlessness, pseudonymity), all of the above deviations can be conducted undetectably, which prevent explicit penalization for deviations.

\subsection{Incentive compatibility}
From a mechanism design perspective, a TFM is effectively an auction where the auctioneer may or may not be truthful. Further, the majority of previous mechanism design literature assumes a trusted auctioneer who honestly executes the mechanism. As such, attempts to design TFM usually focus on the following incentive compatibility properties:
\begin{enumerate}
    \item a TFM is \textit{user incentive compatible} (UIC) if truthful bidding is dominant for users.
    \item a TFM is \textit{myopic miner incentive compatible} (MMIC) if truthful implementation of the mechanism is dominant for miners where only single-round utility is considered (recall \ref{iterated})
    \item a TFM is \textit{off-chain-agreement-proof} (OCA-proof) if no off-chain agreement (OCA) \footnote{Note that throughout this paper we assume off-chain agreements have fulfillment guarantees, as is implied by the literature} between the miner and any number of users can improve joint utility over the outcome from bidding and implementing the mechanism honestly, respectively. We can relax this by denoting TFMs robust against OCAs of the miner and up to $c$ users as $c$-OCA-proof.
    % The weakest form of this is 1-OCA-proof, which only considers collusion between a miner and one user. We will denote $\infty$-OCA-proof, over arbitrary-size cartels, simply as OCA-proof.
\end{enumerate}

\begin{remark}[UIC, MMIC, OCA-proof are incomparable]
No pair of notions yields entailment, which motivates thinking independently about all three, e.g. instead of some unifying property. For most directions this quickly follows from considering first-price, second-price, or posted-price auctions. We will present the proof of one subtle direction to give intuition on the relationships between the notions:
\begin{prop}
OCA-proof does not imply MMIC
\end{prop}
\begin{proof}
Consider an auction where only the highest bid is confirmed, and they simply pay their bid which the miner receives as revenue, except when they are the only bid, in which case they pay nothing and the miner gets nothing. This is OCA-proof since there is no nonzero joint utility for the miner to collude with any unconfirmed user, and the joint utility for the cartel consisting of the miner and the single confirmed user cannot exceed the value $v_i$ of the user. However, this is not MMIC, since in cases where there is only one bid, the miner can inject a fake bid to get revenue nonzero.
\end{proof}
\end{remark}

\begin{remark}(Incentive compatibility properties are naturally motivated) \label{tfmprops}
We see that these canonical properties are also naturally motivated by a view of the TFM as a user-centric transaction-processing service. Possible such desiderata include:
\begin{enumerate}
\item Simplify user bidding, which reduces overpayment, user-side computational costs, and overall user experience
\item Increase network capacity and reduce delays for users
\item Increase network robustness and decentralization properties
\item Allow users to pay for priority inclusion in a block
\item Allow users to pay for other properties of blocks (e.g. transaction ordering or bid escalators \footnote{this refers to a user specifying a more detailed curricula of how much they are willing to pay for their transaction to be included within the $k$ next blocks for each $k$})
\end{enumerate}
Among these, (1) motivates UIC and OCA-proofness as a simple bid should intuitively just be a function of the user's personal utility $v_t$ and not competing bids or off-chain miner offers (e.g. (1) bidding $v_t$ or either bidding a posted price $p$ if $v_t < p$ or abstaining otherwise). Similarly, (2) motivates MMIC, as effective network capacity decreases if miners submit fake bids to manipulate revenue, which is a common way by which miners deviate from the honest protocol. (3) is closely related to mitigations against selfish mining, and additionally can be viewed as a motivation for OCA-proofness, as collusion is a centralizing vector. (4) is a property found in first-price auctions \ref{firstprice} and later ``tipping" in the EIC-1559 mechanism \ref{eip}. (5) is left open and mentioned among other open challenges in the conclusion.
\end{remark}

\section{Three-Way Incentive Compatibility is Impossible}
Sadly, a major result by \cite{chung} implies that getting this combination of incentive compability notions as written is impossible. Specifically:
\begin{theorem} \label{incompatible}
Assuming finite block size, no nontrivial single-parameter, possibly randomized TFM can be both UIC and 1-OCA-proof.
\end{theorem}

To prove this, we will assume the following well known result by Myerson \cite{myerson81}: 
\begin{lemma}
\label{myerson}
% may need stronger form of this
%For deterministic auctions, confirmation is monotone and each confirmed bid pays the minimum amount that still allows it to be confirmed while the other bids are fixed.}{
Consider a single-parameter TFM $(\mathbf{x}, \mathbf{p}, \mu)$ that is UIC. Then the following properties must hold:
\begin{enumerate}
\item We define monotone as follows: Consider some $\mathbf{b} := (b_1,...,b_m)$. An allocation rule $x$ is monotone if for any $\mathbf{b}$, and any $b'_i > b_i$, $x_i(\textbf{b}_{-i}, b'_i) \geq x_i(\mathbf{b_{-i}}, b_i)$. The allocation rule $\mathbf{x}$ is monotone.
\item For any user $i$, bid $b_i$, and other bids $\mathbf{b}_{-i}$, the payment rule is defined as
\begin{enumerate}
\item If the mechanism is non-deterministic
$$ p_i(\mathbf{b}_{-i}, b_i) = b_i \cdot x_i(\mathbf{b}_{-i}, b_i) - \int_{0}^{b_i} x_i(\mathbf{b}_{-i}, t) dt $$
\item Otherwise,
\[ p_i(\mathbf{b}_{-i}, b_i) = \begin{cases} 
          \min\{z \in [0,b_i] : x_i(\mathbf{b}_{-i},z) = 1\} & \text{ if } x_i(\mathbf{b}_{-i}, b_i) = 1 \\
          0 & \text{ if } x_i(\mathbf{b}_{-i}, b_i) = 1  
       \end{cases}
 \]
\end{enumerate}
\item A user $i$'s payment must satisfy the "payment sandwich" inequality
$$ v \cdot (x_i (\mathbf{b}_{-i}, v') - x_i (\mathbf{b}_{-i}, v)) \leq p(\mathbf{b}_{-i}, v') - p(\mathbf{b}_{-i}, v) \leq v' \cdot (x_i(\mathbf{b}_{-i}, v') - x_i(\mathbf{b}_{-i}, v)) $$
Furthermore, for a non-decreasing function $x_i (\mathbf{b}_{-i}, \cdot)$ and $p (\mathbf{b}_{-i}, 0) = 0$, the payment rule is of the unique form presented in (2).
\item Let $f(z)$ be a monotonically decreasing function. If for any $z' \geq z \geq 0$, $z \cdot (f(z') - f(z)) \leq g(z') - g(z) \leq z' \cdot (f(z') - f(z))$ and $g(0) = 0$, then
$$ g(z) = z \cdot f(z) - \int_{0}^{z} f(t) dt $$
\end{enumerate}
\end{lemma}
The idea here is that each user $i$ should only pay the minimum price that confirms their bid. It should be remarked that points (3) and (4) are actually direct corollaries of \cite{myerson81}'s original Myerson lemma. We also opt out of directly proving the deterministic case because it is similar but simpler than the general randomized case and uses the same overarching proof steps. We first prove the following lemma. For notational convenience, we define
$$ \pi_{\mathbf{b}_{-i}}(r) = p_i(\mathbf{b}_{-i}, r) - \mu(\mathbf{b}_{-i}, r) $$
where $\mu(\mathbf{b})$ is the expected miner revenue and $p_i(\mathbf{b})$ is the expected payment of user $i$.
\lem{ \label{ocainequality}
Consider a single-parameter TFM $(\mathbf{x}, \mathbf{p}, \mu)$ that is 1-OCA-proof. Then for any bid vector $\mathbf{b}$, user $i$, and $r,r'$ such that $r < r'$, 
$$ r \cdot (x_i (\mathbf{b}_{-i}, r') - x_i (\mathbf{b}_{-i}, r)) \leq \pi_{\mathbf{b}_{-i}}(r') - \pi_{\mathbf{b}_{-i}}(r) \leq r' \cdot (x_i(\mathbf{b}_{-i}, r') - x_i(\mathbf{b}_{-i}, r)) $$
}{
We first prove the left inequality. Assume for the sake of contradiction that there exists a vector $\mathbf{b}$, a user $i$, and $r < r'$ such that
$$ r \cdot (x_i (\mathbf{b}_{-i}, r') - x_i (\mathbf{b}_{-i}, r)) > \pi_{\mathbf{b}_{-i}}(r') - \pi_{\mathbf{b}_{-i}}(r) $$
Suppose the real bid vector is $(\mathbf{b}_{-i},r)$ and user $i$'s true value is $r$. Then either
\begin{enumerate}
    \item They do not have a side contract, so the miner's expected utility is $\mu(\mathbf{b}_{-i},r)$ and user $i$'s expected utility is $r \cdot (x_i (\mathbf{b}_{-i}, r) - p_i (\mathbf{b}_{-i}, r))$.
    \item They have a side contract, and the miner asks user $i$ to bid $r'$. Then the miner's expected utility is $\mu(\mathbf{b}_{-i},r')$ and user $i$'s expected utility is $r \cdot (x_i (\mathbf{b}_{-i}, r') - p_i (\mathbf{b}_{-i}, r'))$. This implies that their joint expected utility increases by $r \cdot (x_i (\mathbf{b}_{-i}, r') - x_i (\mathbf{b}_{-i}, r)) - (\pi_{\mathbf{b}_{-i}}(r') - \pi_{\mathbf{b}_{-i}}(r))$, which is $>0$ by our assumption. However, this now violates the TFM being 1-OCA-proof, so we have a contradiction.
\end{enumerate}
We now prove the right inequality, which is the same argument. Assume for the sake of contradiction that there exists a vector $\mathbf{b}$, a user $i$, and $r < r'$ such that
$$ \pi_{\mathbf{b}_{-i}}(r') - \pi_{\mathbf{b}_{-i}}(r) > r' \cdot (x_i(\mathbf{b}_{-i}, r') - x_i(\mathbf{b}_{-i}, r)) $$
Suppose the real bid vector is $(\mathbf{b}_{-i},r')$ and user $i$'s true value is $r'$. Then we get the same scenario as before, where the joint expected utility increases by a positive amount when the miner asks the user $i$ to bid $r$ instead through a contract, thus violating the TFM being 1-OCA-proof.  
}
We will use the above lemma and Myerson's lemma to prove the following key lemma:
% \lem{For deterministic single-parameter TFMs, UIC and 1-OCA-proof imply zero miner revenue.}{
% \lem{} {
% }
% }
% \lem{For deterministic single-parameter TFMs with finite block size, 1-OCA-proof and zero miner revenue have imply the TFM is trivial}{
% }

% We can extend this result to randomized TFMs:
\lem{ \label{zerorev}
For any single-parameter TFM, UIC and 1-OCA-proof imply zero miner revenue.
}{
Consider the following quantity, which only differs from $\pi_{\mathbf{b}_{-i}}(r)$ by a constant amount: 
$$ \tilde{\pi}_{\mathbf{b}_{-i}}(r) = p_i(\mathbf{b}_{-i}, r) - \mu(\mathbf{b}_{-i}, r) - (p_i(\mathbf{b}_{-i}, 0) - \mu(\mathbf{b}_{-i}, 0)) $$
This is precisely the currency lost by a contract between user $i$ and the miner by fixing $\mathbf{b}_{-i}$ and changing from $0$ to $b_i$ value. Applying Lemma \ref{ocainequality}, we see that
$$ r \cdot (x_i (\mathbf{b}_{-i}, r') - x_i (\mathbf{b}_{-i}, r)) \leq \tilde{\pi}_{\mathbf{b}_{-i}}(r') - \tilde{\pi}_{\mathbf{b}_{-i}}(r) \leq r' \cdot (x_i(\mathbf{b}_{-i}, r') - x_i(\mathbf{b}_{-i}, r)) $$
By definition, $\tilde{\pi}_{\mathbf{b}_{-i}}(0) = 0$ and because the TFM is UIC and the above expression satisfies the "payment sandwich" in Lemma \ref{myerson}(c), it follows that it must obey Lemma \ref{myerson}(d), that is that
$$ \tilde{\pi}_{\mathbf{b}_{-i}}(r) - b_i \cdot x_i (\mathbf{b}_{-i}, b_i) - \int_0^{b_i} x_i (\mathbf{b}_{-i}, t) dt $$
and because the TFM is UIC, its payment rule must also satisfy
$$ p_{\mathbf{b}_{-i}}(r) - b_i \cdot x_i (\mathbf{b}_{-i}, b_i) - \int_0^{b_i} x_i (\mathbf{b}_{-i}, t) dt $$
and therefore,
$$ \tilde{\pi}_{\mathbf{b}_{-i}}(r) = p_i(\mathbf{b}_{-i}, r) - \mu(\mathbf{b}_{-i}, r) - (p_i(\mathbf{b}_{-i}, 0) - \mu(\mathbf{b}_{-i}, 0)) = p_i(\mathbf{b}_{-i}, r) $$
which further implies that $\mu(\mathbf{b}_{-i}, r) = \mu(\mathbf{b}_{-i}, 0) - p_i(\mathbf{b}_{-i}, 0)$ is a constant when $\mathbf{b}_{-i}$ is held fixed.
\\\\
To finish off the proof, suppose for the sake of contradiction that there exists a single-paramter TFM that is UIC and 1-OCA-proof and has non-zero miner revenue. Then there exists a bid vector $\mathbf{b}^{(0)} = (b_1,...,b_m)$ such that $\mu(\mathbf{b}^{(0)}) > 0$. If we consider the sequence of bid vectors constructed such that for $i \in [m]$, $\mathbf{b}^{(i)} = (0,...,0,b_{i+1},...,b_m)$ and $\mathbf{b}^{(m)} = 0$. Recall that for a fixed $\mathbf{b}_{-i}$, we showed that under the current assumptions, $\mu(\mathbf{b}_{-i}, \cdot)$ is independent of user $i$'s bid. So in our sequence of bid vectors, $\mathbf{b}^{(i)} = \mathbf{b}^{(i-1)}$ and therefore $\mathbf{b}^{(0)} = \mathbf{b}^{(m)}$. But because user $i$ can only pay at most their bid, $\mu(\mathbf{b}_{i}) \leq |\mathbf{b}_{i}^{(m)}|_1 = 0$, which contradicts our assumption that $\mu(\mathbf{b}^{(0)}) > 0$. Thus, by contradiction, a single-parameter TFM that is UIC and 1-OCA-proof and has zero miner revenue.
}
Now we restrict ourselves to finite block sizes, and use Lemma \ref{zerorev} to show our impossibility theorem.
\lem{ \label{trivial}
For randomized single-parameter TFMs with finite block size, 1-OCA-proof and zero miner revenue, the TFM is the trivial mechanism that always pays the miner nothing and never confirms any transactions.
}{
Let $B$ denote the finite block size. Suppose for the sake of contradiction that there exists a non-trivial TFM $(\mathbf{x}, \mathbf{p}, \mu)$ that satisfies both UIC and 1-OCA-proof. So there exists a bid vector $\mathbf{b} = (b_1,...,b_m)$ and a user $i$ such that $x_{i}(\mathbf{b}) > 0$. Now for some positive value $p > 0$, let $N > \frac{B \cdot (b_i + p)}{x_i(\mathbf{b}) \cdot p}$ be some large integer, and consider another bid vector $\mathbf{b}' = (b_1,...,b_m,b_{m+1},...,b_{m+N})$ where $b_j = b_i + p$ for all $j \in [m+1,m + N]$. If this is the real bid vector and each user bids truthfully, then there exists a user $j \in [m+1,m+N]$ who bids $b_j$ and is included with probability $\leq \frac{B}{N} < \frac{x_i(\mathbf{b}) \cdot p}{b_i + p}$.

Now consider some user $j$. Under the honest mechanism, their joint utility with the miner (who gets 0 revenue by Lemma \ref{zerorev}) is 
$$ \leq b_j \frac{B}{N} < b_j \frac{x_i(\mathbf{b}) \cdot p}{b_i + p} = x_i(\mathbf{b}) \cdot p $$
But if the miner and user $j$ form a contract where user $j$ changes their bid from $b_j$ to $b_i$ and the miner pretends the real bid vector is $\mathbf{b}$ where $b_i$ actually comes from user $j$, their joint utility is
$$ (v_j - b_i) \cdot x_i (\mathbf{b}) = p \cdot x_i(\mathbf{b}) $$
and therefore 1-OCA-proof is violated. Thus, by contradiction, the TFM must be the trivial mechanism.
}
The proof of Theorem \ref{incompatible} directly follows from Lemma \ref{zerorev} and Lemma \ref{trivial}. $\square$

\section{Weaker Forms of Incentive Compatibility} \label{weakerIC}
Perhaps this incompatibility suggests that UIC and 1-OCA-proof are too strong. Short of ensuring zero payoff for users and miners being dishonest, where payoff is the utility increase from deviating compared to being honest, we can aim for the following:
\begin{enumerate}
    \item Bound the (worst-case or average) payoff from deviating
    \item Exploit the fact that fake or overbid transactions that are included but not confirmed cannot be retracted in many blockchains and thus run the risk of being confirmed at a later time and thus impose expected cost onto the miner, user, or cartel that proposed them (\ref{iterated}). Thus, we can try to incorporate this cost into utilities and prove weaker forms of incentive compatibility on these utilities
    % that account for unconfirmed fake transactions that may be fulfilled in the future, imposing negative utility on their strategic bidders
\end{enumerate}
\subsection{Nearly incentive compatible}
Towards (1), we extend \cite{lavi}'s notion of \textit{nearly incentive compatible} for users to nearly-IC, nearly-MMIC, nearly-$c$-OCA-proof.
\begin{definition} (discount ratio)
    A user $i$'s payoff from deviating can be modeled as their percent savings on payment or \textit{discount ratio}
    \eq{
    \delta_i^{discount} (v_i, \bn) \ce\begin{cases}
        \frac{p^t - p^*}{p^t} &\text{if } v_i \ge p^*(\bn) \\
        0 &\text{else}
    \end{cases}
    }
    where $p^*$ is the minimal price a user can pay to have their bid confirmed and $p^t$ is the price a user pays if they bid truthfully. Note that the 0 branch of the function simply formalizes the intuition that payoff can only be nonzero in cases where the user has any chance of getting positive utility. This also aligns with the fact that $0 \le \delta_i \le 1$.
\end{definition} \cite{lavi}
\begin{definition} (average and max discount ratio)
    Suppose true values are drawn iid from a distribution $F$ on $\R_+$.
    Then, the expected payoff for a user deviating assuming all other users are honest can be modeled the \textit{average discount ratio}
    \eq{
    \Delta_n^{discount, avg} \ce \E_{\vv \sim F} \delta_1^{discount}(v_1, \vv_{-1})
    }
    where WLOG by symmetry the choice of bidder 1 is arbitrary.

    We can also reason about the \textit{max discount ratio}
    \eq{
    \Delta_n^{discount, max} \ce \E_{\vv \sim F} \max_j \delta_j^{discount}(v_j, \vv_{-j})
    }
\end{definition} \cite{lavi}
We can adapt these to a generalized notion regarding additive payoff instead of percentage discount on payments that can also be applied to MMIC and OCA-proof and is a bit more natural to interpret being directly in terms of utility.
\begin{definition} (strategic payoff)
    For a user, miner, or cartel with up to $c$ users, their \textit{strategic payoff} is
    \eq{
    \delta(V, \bn) \ce \max\prn{u^* - u^t, 0}
    }
    where, $V \sbe \vv$ consists of the values, if any, of the users in the colluding group. (i.e. $\set{v_i}, \emptyset, \set{v_{i_1},...,v_{i_c}}$ for a user, a miner, and a cartel of $c$ users, respectively), and utility refers to joint utility, e.g. a sum of possibly $\sum_{v_j \in V} v_j - p_j$ for any users and $r$ for a miner.
\end{definition}
\begin{definition} (average and max strategic payoff) (ours)
    With the same assumptions on $F$ and other users as in discount ratio, let the \textit{average strategic payoff} for a user, miner, or cartel with up to $c$ users be
    \eq{
    \Delta_n^{avg} \ce \E_{\vv \sim F} \delta(V, \vv \setminus V), V \ce \set{v_1,...,v_c}
    }
    where WLOG by symmetry the choice of set V (up to size for cartels) is arbitrary.

    We can also reason about the \textit{max strategic payoff}
    \eq{
    \Delta_n^{max} \ce \E_{\vv \sim F} \max_{V'} \delta(V', \vv \setminus V')
    }

    For clarity, we can further superscript these by users, miners, and cartels with up to $c$ users (e.g. $\Delta_n^{av, users}$,\\$\Delta_n^{max, \le c \text{ }cartels}$).
\end{definition}

Note that as \cite{lavi} implicitly does, we will assume
% that all TFMs satisfy $u^* \ge u^t > 0$ for all users and
$p^t \ge p^* > 0$ for all confirmed users.

Generalizing \cite{lavi} \cite{yao}'s analysis on near incentive compatibility for monopolistic auctions (theorem and scenario presented and discussed below \ref{monopolistic}), we can define:
\begin{definition}
    A TFM is \textit{nearly UIC} if for any $F$, the average payoff ratio for users satisfies $\Delta_n^{avg, users} \to 0$ with $\Delta_n^{avg, users} = O(1/n^\beta)$ for some constant $\beta > 0$ independent of $F$.

    The respective expressions using $\Delta_n^{avg, users}, \Delta_n^{avg, \le c \text{ } cartels}$ give the definitions of \textit{nearly MMIC} and \textit{nearly c-OCA-proof}, respectively.
\end{definition}

It is also worth noting that this is not the only formulation of approximate incentive compatibility via bounded payoffs (e.g. additive vs. multiplicative).
% We use a multiplicative formulation to more directly extend analysis from \cite{yao} and \cite{lavi}, which rely on properties like $\delta \in [0,1]$, but additive notions (e.g. $\delta \ce u^*-u^t$ may be more natural in some scenarios).
The choice of $O(1/n^\beta)$ is also somewhat arbitrary, and our notion may be regarded as $1/n^\beta$-near incentive compatibility as a special case of $f(n)$-near incentive compatibility.

% \label{monoNear}
% \subsubsection{Example: monopolistic auctions are nearly UIC}
% \begin{theorem} (Nearly Bayesian-Nash Incentive Compatibility)
%     \begin{enumerate}
%     \item For $F$ with bounded support, $\Delta^{discount, max}_n \to 0$ and in particular $\Delta^{discount, max}_n = O(1/\sqrt{n})$.
%     \item For any $F$, $\Delta^{discount, avg}_n \to 0$ and in particular $\Delta^{discount, avg}_n = O(1/\sqrt{n})$.
%     \end{enumerate}
    
%     From \cite{yao}, when discussing asymptotic bounds, we will allow the constant in $O(\cdot)$ to depend on $F$.
% \end{theorem} \cite{lavi} \cite{yao}

% An additional advantage is that our definition is invariant to absolute user values for UID, and only depends on net utilities, which is a natural property in studying incentives to deviate.
% To motivate this formulation, note that we can compare it to discount ratio

\subsection{Weak incentive compatible}
Towards (2), following the proof of the above impossibility theorem, \cite{chung} introduces the notion of weak incentive compatibility, where unconfirmed transactions are penalized with discount rate $\gamma$ \footnote{One natural perspective on this is that $\gamma$ is the probability of the unconfirmed transaction being confirmed in a future block}. We adapt an equivalent version of their definition:
\begin{definition} ($\gamma$-strict utility)
    Let a player, miner, or cartel's utility under a strategy be $u$, which does not include unconfirmed bids.
    For each unconfirmed bid $i$ let $u_i$ be the utility resulting from that bid being confirmed. Note that this may include miner revenue coming from the bid.
    % Let $\bb, \vv$ be their bids and values for any unconfirmed fake transactions they submit under the strategy.
    Then, their \textit{$\gamma$-strict utility} is $u + \gamma \sum_i \min(u_i, 0)$ for $\gamma \in (0,1]$.
\end{definition}
\begin{definition} ($\gamma$-weak incentive compatibility)
    For $\gamma \in (0,1]$, let $\gamma$-weak UIC be UIC under $\gamma$-strict utility. Define $\gamma$-weak MIC, $\gamma$-weak $c$-OCA-proof respectively for MIC, $c$-OCA-proof.
\end{definition}
\begin{remark}
Note that since each additional term is non-positive, $\gamma$-strict utility $u_\gamma$ must satisfy $u' \le u$, so trivially each incentive compatibility notion implies its $\gamma$-weak counterpart for all $\gamma \in (0,1]$. Similarly, $u_\gamma$ is monotonically decreasing in $\gamma$, so in particular the weakest notion of 1-weak utility can be used as a sufficient condition to prove that a TFM is not $\gamma$-weak incentive compatible for any $\gamma \in (0,1]$.
\end{remark}

\subsection{Properties of weak incentive compatible TFMs} \label{weakprops}
It turns out that weak TFM is still rather strong; as a start, satisfying all three notions of weak incentive compatibility already prescribes randomness and unconfirmed transactions as mandatory.

The following notion turns out to be a useful and natural assumption:
\begin{definition}
    A TFM is \textit{2-user-friendly} if there exists a bid vector such that it confirms at least two bids \cite{chung}
\end{definition}

\begin{theorem} (Necessity of randomness for weak UIC and 2-weak-OCA-proof) \cite{chung}
    A deterministic and 2-user-friendly TFM with finite block size cannot be both weak UIC and 2-weak-OCA-proof.
\end{theorem}

\begin{theorem} (Necessity of unconfirmed transactions for weak UIC and 1-weak-OCA-proof) \cite{chung}
    A nontrivial (possibly randomized) TFM that always confirms all transactions cannot be both weak UIC and 1-weak-OCA-proof.
\end{theorem}

We will omit the extensive proofs, but they can be found in \cite{chung}.

\section{Beyond Incentive Compatibility}
Beyond incompatibility, we may still be interested in other notions that help us characterize the tradeoffs of a TFM: we give one example here.

% \textit{stronger} notions of incentive compatibility.
In the spirit of \ref{tfmprops}, it is possible that the benefits of reducing fake transactions via MMIC far outweigh those of simple user fees via UIC \footnote{a concrete example might be if we can simpify bidding by providing access to a first-price-auction bidding oracle, but storage space for total transactions is limited}, and in particular we want to impose a significant penalty on the miner for any fake transaction included. Since incentive compatibility applies for zero utility, even MMIC and UIC cannot give any guarantees on the magnitude of a protection against fake transactions, and at best imply that a miner or user is not encouraged to do so (but e.g. may very well be able to flood the network with fake transactions without any penalty).

Towards this end, we adapt an equivalent version of Roughgarden's \cite{roughgarden20}\cite{roughgarden21} notion of $\gamma$-costly for punishing fake transactions:
\begin{definition} ($\alpha$-costly)
A TFM is $\alpha$-costly if each confirmed fake transaction decreases its bidder's utility by at least $\alpha$.
\end{definition}
We usually define this in relation to a $\ep$-costly baseline, e.g. any UIC and MMIC TFM. In doing so, we assume a marginal cost $\ep>0$ incurred for all fake transactions, even those leading to no utility penalty. In real-world scenarios, this can be seen as added orphan risk to the miner due to blocks with more data taking longer to propagate. Thus, this notion is usually invoked when a given TFM achieves $\alpha$-costly with $\alpha \gg \epsilon$, as we will see with \ref{eip}.

\begin{remark}
    Note that this is thematically distinct from $\gamma$-strict utility from \cite{chung}, which concerns quantifying the total utility impact of unconfirmed transactions to relax MMIC. This notion instead quantifies the per-transaction of confirmed (fake) transactions to strengthen MMIC.
\end{remark}

% \textbf{Remark} Why do users need to pay a nonzero amount at all to have their transaction included? The scarcity of the underlying token Indirect payment methods include, but it seems natural that users that use the network should pay more than tokenholders (a la tolls and subway costs vs. tax-maintained freeways)

\section{Classical TFMs} \label{classical}
We will use the above incentive compatibility notions and their weaker forms to analyze why classical mechanisms fall short of being good TFMs.

\subsection{First-price auction} \label{firstprice}
\begin{algo}
The \texttt{First-Price Auction} mechanism parameterized by the block size $B$ behaves as follows:
\begin{enumerate}
    \item \textbf{I}: Include the $B$ highest bids $b_1 \ge ... \ge b_B$, breaking ties arbitrarily. If the block is not full, we treat empty slots as bids of 0.
    \item \textbf{C}: Confirm all bids.
    \item \textbf{P, R}: The $i$th confirmed bid pays $b_i$ to the miner, and unconfirmed bids pay nothing. 
\end{enumerate}
\end{algo}

\prp{The first-price auction is not UIC, or even 1-weak UIC}{
% (2) MMIC (3) OCA-proof}{
Suppose $v_1 = 10, b_2 = 2, b_3 = 1$. The first bidder can save at least 7 by bidding e.g. $3$ and still having their bid confirmed. Note that no fake transactions are necessary here so 1-strict utility is the same as regular utility.
}

\prp{The first-price auction is MMIC}{
Since miner revenue is the sum of the bids, it is best response is to truthfully pick the highest ones.
}

\prp{The first-price auction is OCA}{
Consider a cartel of any number of users and the miner. The bids of the users in the cartel do not affect joint utility as their payments go to the miner, so we can arbitrarily set their bids to be their values (i.e. honest). Then, it remains to pick the $B$ bids to include: including a cartel user increases the joint utility by their value and including a non-cartel user increases the joint utility by their bid. Thus, to maximize joint utility the miner should honestly include the top $B$ bids to be confirmed.
}

% (3) Consider the joint utility of a cartel. Clearly if no bidder's confirmation status changes, the sum of confirmed bidders' utility with the miner is precisely 0, and unconfirmed bidders' utility also have utility 0 and contribute no revenue to the miner.

% Unfortunately, first-price auctions don't even get close to UIC:
% \prp{First-price auctions are (1) not nearly UIC and (2) not 1-weak UIC}{
% (1) Consider $F$ bounded such that all other users

% (2) The example above for regular UIC also suffices to show this, as no fake transactions are necessary to be strategic so 1-strict utility is the same as regular utility.
% }

\subsection{Second-price auction} \label{secondprice}
\begin{algo}
The \texttt{Second-Price Auction} mechanism is parameterized by the block size $B$. It behaves as follows: (differences from first-price bolded)
\begin{enumerate}
    \item \textbf{I}: Include the $B$ highest bids $b_1 \ge ... \ge b_B$, breaking ties arbitrarily. If the block is not full, we treat empty slots as bids of 0.
    \item \textbf{C}: Confirm all bids.
    \item \textbf{P, R}: All confirmed bids pay $\mathbf{b_{B+1}}$ to the miner, and unconfirmed bids pay nothing. 
\end{enumerate}
\end{algo}

Second-price auctions are famously UIC, but when the auctioneer is untrusted, fall to MMIC and OCA-proof:
\prp{Second-price auctions are not even 1-weak MMIC or 1-weak 1-OCA-proof.}{
Suppose $\bb = (10,9,8,3), k = 3, B = 4$. The miner gets revenue $3(3)=9$ honestly but by injecting a fake bid of $7$ can raise this to $3(7)=21$. 
\footnote{Generally, injecting a fake bid between $b_k,b_{k+1}$ is the dominant strategy}
Even under 1-strict utility we have $21-7 \ge 9$, so second-price auctions are not 1-weak MMIC.

Similarly, the miner can form a cartel with the fourth bidder to increase their bid to 7, raising their joint utility from $3(3)+0=9$ to $3(7)+0=21$, and even under 1-strict utility this yields (depending on $v_4$) $21+\min(7+v_4-7,0) \ge 21 \ge 9$, so second-price auctions are not 1-weak 1-OCA-proof.
}

\subsection{Monopolistic auction} \label{monopolistic}
Second-price auctions quickly motivate:
\begin{algo}
The \texttt{Monopolistic Auction} mechanism is parameterized by the block size $B$. It behaves as follows: (differences from second-price bolded)
\begin{enumerate}
    \item \textbf{I}: Include the $\mathbf{k^* \ce \max_{k} kb_k}$ highest bids $b_1 \ge ... \ge b_{k^*}$, breaking ties arbitrarily. If the block is not full, we treat empty slots as bids of 0.
    \item \textbf{C}: Confirm all bids.
    \item \textbf{P, R}: All confirmed bids pay $b_{k^*}$ to the miner, and unconfirmed bids pay nothing. 
\end{enumerate}
\end{algo}

\prp{Monopolistic auctions are not 1-OCA-proof}{
Suppose $\bb = (10,9,7,3), k = 3, B = 4$. The miner gets revenue $3(7)=21$ honestly but by forming a cartel with the third bidder to increase their bid to 8, they can raise their joint utility from $3(7)+v_3-7=14+v_3$ to $3(8)+v_3-8=16+v_3$, which does not change under 1-strict utility.
}

Although, a result from \cite{yao} \cite{lavi} does give (here stated using their original discount-ratio-based definition) that
\begin{prop} Monopolistic auctions are nearly UIC, satisfying:
\begin{enumerate}
    \item For $F$ with bounded support, $\Delta^{discount, max}_n \to 0$ and in particular $\Delta^{discount, max}_n = O(1/\sqrt{n})$.
    \item For any $F$, $\Delta^{discount, avg}_n \to 0$ and in particular $\Delta^{discount, avg}_n = O(1/\sqrt{n})$.
\end{enumerate}
When discussing asymptotic bounds, we will allow the constant in $O(\cdot)$ to depend on $F$.
\end{prop}
\begin{proof}
We omit the proof of (1), which can be found in \cite{yao}, but will prove (2) given (1) to give some intuition on the probabilistic nature of the definition.

Fix a $0 < \ep < 1$ and a distribution $F$.

First, pick a $D > 0$ where $F(D) > 1-\ep/3$. Suppose we are sampling $n$ iid $v_1,...,v_n \sim F$.
Let $X_n$ be a random variable denoting the number of $v_i$'s that fall in $[0,D]$. By the Law of Large numbers, there exists $N_1$ where $n \ge N_1$ ensures that $\P(X_n \le (1-\ep/2)n) < \ep/6$.

Then, let $G$ be $F$ restricted to $[0,D]$, and we will crucially apply (1) to get $N_2$ where $m \ge N_2$ ensures $\Delta_m^{discount, avg}(G) \le \Delta_m^{max} < \ep/3$.

Now we are ready to prove the result. Pick $n \ge \max(N_1,N_2)$ and fix $v_1,...,v_n \sim F$. It suffices to show that $\Delta_n^{discount, avg}(F) \le \ep$. Recall that $\Delta_n^{discount, avg} \in [0,1]$. Then, we can write
\eq{
\Delta_n^{discount, avg}(F) &= \P(X_n \le (1-\ep/2)n)\Delta_n^{discount, avg}(F) + \P(X_n > (1-\ep/2)n)\Delta_n^{discount, avg}(F) \\
&\le (\ep/6)(1) + \sum_{m > (1-\ep/2)n}\P(X_n = m)[\Delta_n^{discount, avg}(F)] \\
&\le \ep/6 + \sum_{m > (1-\ep/2)n}\P(X_n = m)[\frac{m}{n}\Delta_n^{discount, avg}(G) + \frac{n-m}{m}(1)] \\
&\le \ep/6 + \sum_{m > (1-\ep/2)n}\P(X_n = m)[(1)(\ep/3) + (\ep/2)(1)] \\
&\le \ep/6 + \ep/3 + \ep/2 = \ep
}
where in step 3 we apply the definition of $\Delta_n^{discount, avg}$, and in step 4 we use the facts that $\frac{m}{n} \le 1, \frac{n-m}{m} \le \ep/2$
\end{proof}
% \prp{Monopolistic auctions are MMIC}{

% }

\subsection{Posted-price auction} \label{posted}
\begin{algo}
The \texttt{Posted-Price Auction} mechanism is parameterized by the block size $B$ and a posted price $p$. It behaves as follows:
\begin{enumerate}
    \item \textbf{I}: Include any $B$ bids (that equal $p$). If the block is not full, we treat empty slots as bids of 0.
    \item \textbf{C}: Confirm all bids (that equal $p$).
    \item \textbf{P, R}: All confirmed bids pay $p$ to the miner and unconfirmed bids pay nothing.
\end{enumerate}
\end{algo}

\prp{Posted-price auctions are not 1-weak-OCA-proof}{
Suppose $p=10,\bb = (0,0,0,0), \vv = (5,0,0,0)$. The miner can form a cartel with user 1 and have them bid the posted $10$, which increases their joint utility from $0$ to $5-10+10=5$, which does not change under 1-strict utility.
}

\section{Towards Incentive Compatible TFMs}
\subsection{EIP-1559 mechanism}
\label{eip}
Now we will examine our first arguably robust mechanism, a combination of posted-price and first-price, which was proposed in \cite{eip1559} and later implemented in the Ethereum mainnet as the successor to the first-price auction, where it now runs.

\begin{algo}
The \texttt{EIP-1559 Mechanism} is parameterized by the block size $B$ and a posted price $p$. It behaves as follows: (differences from second-price bolded)
\begin{enumerate}
    \item \textbf{I}: Include highest $B$ bids $b_1 \ge ... \ge b_B \mathbf{\ge p}$, breaking ties arbitrarily. If the block is not full, we treat empty slots as bids of 0.
    \item \textbf{C}: Confirm all bids. Bids must be $\ge p$.
    \item \textbf{P}: The $i$th confirmed bid pays $b_i$, and unconfirmed bids pay nothing.
    \item \textbf{R}: For the $i$th confirmed bid, the miner gets $\mathbf{b_i - p}$ \textbf{and the remaining $p$ is burned}.
\end{enumerate}
\end{algo}
Intuitively, $p$ is usually denoted the \textit{base fee} and the remaining payment $b_i-p$ the \textit{tip}.

\prp{EIP-1559 is MMIC}{
The miner's revenue is exactly the sum of the bids they include, minus a constant $Bp$, so optimizing revenue is always equivalent to picking the highest bids honestly, and fake transactions are a strict loss of $p$ and never increase utility.
}

In fact, it explicitly punishes fake transactions:
\prp{EIP-1559 is $(p+\ep)$-costly}{
This is easy to see, as any fake transaction incurs at least the posted price $p$ along with $\ep$.
}

\prp{EIP-1559 is OCA-proof}{
Consider a cartel of any number of users and the miner. The bids of the users in the cartel should be below $p$ if their values are below $p$ (as confirming a bid by such a user would contribute negative joint utility) and at least $p$ if their values are at least $p$ so as to be an inclusion candidate. Beyond this, the tips do not matter as they go to the miner, so we can arbitrarily set their bids to be their values (i.e. honest). Then, it remains to pick the $B$ bids to include. Ignoring users with bids below $p$, including a cartel user $i$ increases the joint utility by $v_i-p \ge 0$, and including a non-cartel user $i$ identically increases the joint utility by their tip $v_i-p \ge 0$. Thus, to maximize joint utility the miner should honestly include the top $B$ bids to be confirmed (for all at least $p$).
}

But in accordance with incompatibility, incentive compatibility must fail somewhere. Here, it fails rather elegantly, not being UIC but lending itself to an obvious bid (as desired under \ref{tfmprops}) except under a specific regime of high demand:
\prp{EIP-1559 is typically UIC-like in the sense that there is a Nash of paying the base fee whenever at most $B$ users have value strictly greater than $p$. \footnote{Note that \cite{roughgarden20} actually gives an alternate definition of UIC as there exists a symmetric ex post Nash equilibrium that is best response for users, which does not necessarily have to be honest wrt their value. Under that definition, EIP-1559 is explicitly UIC outside the high demand regime. However, we remark that in the spirit of simplifying fees for users, an honest Nash is in practice very similar to what is usually just a slightly underbid Nash, and further \cite{chung} and \cite{roughgarden20} point out that one can always convert between mechanisms satisfying these two definitions of UIC easily through what is known as the revelation principle.}
% When more than $B$ users have value strictly reatrer than $p$, it reverts to a first-price auction \cite{roughgarden20} \cite{roughgarden21}
}{
We claim that for user $i$, bidding $min(p,v_i)$ is a Nash. To see this, under the assumptions and all other users using the same strategy, we can bid the absolute minimum $p$ and be confirmed with the utility $v_i-p > 0$ since there are not enough users to fill up the block, so this is best response, as bidding more leads to monotonically greater payment and bidding less leads to 0 utility.
}
Of course, on that regime it completely fails to be UIC.
\prp{EIP-1559 is not 1-weak UIC}{
Suppose $p=5,B=3, \vv = (16,10,10,10), b_2=b_3=5$. The first bidder can save at least 5 by bidding e.g. $11$ and still having their bid confirmed as one of the top three, rather than bidding truthfully ($b_1=16$). Note that no fake transactions are necessary here so 1-strict utility is the same as regular utility.
}

Intuitively, if more than $B$ users have value strictly greater than $p$, then for all users it is preferable to bid at least $p$ (but less than their value $v_i$) to have a chance at getting confirmed and achieving nonzero utility. Thus, the situation is closer to a first-price auction where all confirmed users are penalized a flat amount $p$.

Alternate characterizations of the regime where UIC fails are ``low $p$" or ``high demand," depending on assumptions on the dynamics of user values and posted prices. Notably, \cite{eip1559} \cite{roughgarden20} motivate this design in economic terms, where the setting and updating of $p$ attempts to capture current demand, so UIC is seen as mostly valid up to a good setting of $p$. This leads to various open problems such as (1) optimal ways to set and update $p$ to capture demand based on incoming bids, towards minimizing the failure regime and (2) optimal ways to alleviate the lack of UIC during the failure regime, possibly trading off against other properties.

Related to (2), we note some variations of the EIP-1559 mechanism and results explored by \cite{roughgarden20} \cite{roughgarden21}, which give us a better sense of the design space and relevant tradeoffs:

% \prp{We can sacrifice anonymity to get UIC}{
% Consider EIP-1559 where we e.g. use lower user ID and no payments (could never work in permissionless settings)
% }
% However, this somewhat breaks the assumptions

Suppose we wanted to capture the value of the burned payments by passing them to the miner as revenue. Unfortunately:
\prp{Passing the burned values to the miner is equivalent to a first-price auction, and is not 1-weak-UIC ever}{
The miner and users can essentially take the auction off-chain, where the users communicate their bids off-chain to the miner, and then each user $i$ bids $p$ on-chain and then transfers $b_i-p$ off-chain to the miner, resulting in an effective bid of $b_i$ paid to the miner. This is the same formulation as the first-price auction above \ref{firstprice}, and by our previous result is not 1-weak-UIC.
}

In fact, passing any amount of the payments is meaningless, as this informal note from \cite{roughgarden20} \cite{roughgarden21} shows.
\begin{remark}
    Burning a $\alpha$ fraction of the base fee is economically equivalent to having no base fee at all, since in a scenario where originally the base fee would converge to $p*$, burning a $0<\alpha<1$ fraction of it and passing the rest to the miner as revenue would lead to the base fee converging to $2p^*$, with the unburnt fees redistributed between the mienr and colluding users through an OCA.
\end{remark}

Having conceded that full fee burning is needed, we can also motivate the existence of a base fee. Suppose we ran EIP-1559 with $p=0$, and simply burned all payments in an attempt to preserve incentive-compatibility for the miner. Denote this setup by a \textit{fee-burning first-price auction}, citing our earlier result that EIP-1559 behaves like a first-price auction on tips, and noting its similarity to the canonical first-price auction \ref{firstprice}. Sadly:
\prp{Fee-burning first-price auctions are not 1-weak OCA-proof \cite{roughgarden20} \cite{roughgarden21}}{
The obvious OCA is to move all payments off-chain and conduct a first-price auction with no fee burning. For example, if $v_1 = 10$ the miner can form a cartel with user 1, have them bid and burn 0 and include their transaction, which yields a payoff of 10, e.g. in exchange for an off-chain payment of 5. There is no need for fake transactions so the 1-strict utility is the same as regular utility.
}

Interestingly, we can trade off UIC in the failure regime:
\begin{algo}
The \texttt{Tipless EIP-1559 Mechanism} is parameterized by the block size $B$ and a posted price $p$. It behaves as follows: (differences from standard EIP-1559 mechanism in bold)
\begin{enumerate}
    \item \textbf{I}: Include highest $B$ bids $b_1 \ge ... \ge b_B \mathbf{\ge p}$, breaking ties arbitrarily. If the block is not full, we treat empty slots as bids of 0.
    \item \textbf{C}: Confirm all bids. Bids must be $\ge p$.
    \item \textbf{P}: The $i$th confirmed bid pays $b_i$, and unconfirmed bids pay nothing.
    \item \textbf{R}: For the $i$th confirmed bid, the miner gets $b_i - p$ and the remaining $p$ is burned.
\end{enumerate}
\end{algo}

\cite{roughgarden20}\cite{roughgarden21} shows that analogously to the standard EIP-1559 mechanism:
\begin{prop}Tipless EIP-1559 is MMIC and UIC, but not 1-weak OCA-proof
\end{prop}
which effectively allows to trade off UIC against OCA-proofness in the failure regime, depending on what we prioritize in the \ref{tfmprops} sense. (e.g. simple fees vs. removing centralization vectors).

% Second-price EIP

\subsection{Burning second-price auction} \label{burning}
We conclude by briefly mentioning a mechanism proposed by \cite{chung} specifically to show that three-way $\gamma$-weak incentive compatibility for any $\gamma \in (0,1]$ is possible. Following their necessity theorems, it indeed uses burning and randomness to achieve this result.

\begin{algo}
The \texttt{Burning Second-Price Auction} mechanism is parametrized by the block size $B$, a max cartel size $c$, a discount factor $\gamma \in (0,1]$, and a number of included bids $k \le B$ and number of included but unconfirmed bids $k'$ satisfying $k' = B-k, 1 \le k \le \lfloor \gamma k / c \rfloor$. It behaves as follows:
\begin{enumerate}
    \item \textbf{I}: Include the $B$ highest bids $b_1 \ge ... \ge b_B$, breaking ties arbitrarily. If the block is not full, we treat empty slots as bids of 0.
    \item \textbf{C}: Confirm a random subset of $\lfloor \gamma k / c \rfloor$ included bids with on-chain randomness.
    \item \textbf{P}: All confirmed bids pay $b_{k+1}$, and unconfirmed bids pay nothing. 
    \item \textbf{R}: The miner gets $\gamma(b_{k+1}+...+b_B)$, and all remaining payment is burned.
\end{enumerate}
\end{algo}

\begin{prop} For any $\gamma \in (0,1]$, the burning second-price auction is $\gamma$-weak UIC, $\gamma$-weak MMIC, and $\gamma$-weak OCA-proof.
\end{prop}
% TODO
% }

\subsection{Summary}
We summarize the above results in the following table:
\begin{table}[H]
    \centering
    \begin{tabular}{rlll} \toprule
         & UIC & MMIC & OCA-proof \\ \midrule
        First-price & & \checkmark & \checkmark \\
        Second-price & \checkmark & & \\
        Monopolistic & nearly & \checkmark &  \\
        Posted-price & \checkmark & $\checkmark^*$ & \\ \midrule
        EIP-1559 (low demand) & $\checkmark^{**}$ & $\checkmark^*$ & \checkmark \\
        EIP-1559 (high demand) & & $\checkmark^*$ & \checkmark \\
        Tipless EIP-1559 (low demand) & $\checkmark^{**}$ & $\checkmark^*$ & \checkmark \\
        Tipless EIP-1559 (high demand) & \checkmark & $\checkmark^*$ & \\
        Burning second-price & weak & weak & weak \\
        \midrule
    \end{tabular}
    \caption{Comparison of properties of TFMs}
    \label{tab:tfm_table}
\end{table}
$\mbox{*}$ $(\ep+\alpha)$-costly for a constant parameter $\alpha$

$\mbox{**}$ UIC-like, i.e. the paying the base price is a Nash

\section{Conclusion and Future Work}
The transaction mechanism design problem aims to replace trusted third parties with a coordination game involving incentives. The tasks of coordinating a distributed system, setting prices to match demand and transaction utility, and ultimately allocating assymetrically-valued transactions into vaguely-equal-effort-to-verify blocks thus fall to the users and miners, and by proxy, the mining and mechanism designer who writes those canonical roles for the users and miners to figure out. All of these essentially boil down to figuring out \textit{which users' transactions go where in the ledger} and \textit{which miner is responsible for verifying each transaction}, which motivates the separation of fields that have converged to the mining game and the transaction fee auction respectively.

In analyzing these two sets of problems, we have discovered a menagerie of attacks and misalignments, some of which are provably intractable \ref{incompatible}. But even under intractability, it is clear that we can make significant progress \ref{eip} \ref{burning} if we carefully define \ref{weakerIC} feasible objectives and analyze \ref{weakprops} design elements that can get us there. At the very least, it is clear that intuitive notions of transaction fees mining and mechanism design \ref{honest} \ref{firstprice} fall far short of what we can achieve, so formal analysis is both necessary and urgent.

We leave the reader with a few open questions to guide future research, collected broadly from ideas in the literature:
\begin{enumerate}
    \item Are the currently investigated undercutting and fee sniping attacks resolved completely by the nLockTime protocol, and if not, what new types of equilibrium behavior arise?
    \item Is $\gamma$-weak the strongest notion of incentive compatibility under which we can simultaneously achieve some version of UIC, MMIC, and OCA-proofness?
    \item Where is the provable limit for incentive-compatibility exactly: can we prove tighter impossibility bounds (e.g. with properties beyond zero miner revenue) or show examples of TFMs with stronger guarantees?
    \item Can we more precisely characterize the intuition of unconfirmed fake transaction risk as exploited in $\gamma$-strict utility? Are there more expressive notions of this expectation than a linear sum, e.g. discounted utility horizon, superlinear risk with respect to the size of the transaction?
    \item To what extent is distance-from-equilibrium analysis à la nearly incentive compatible productive? Are the assumptions (fixed distributions, iid values) reasonable under real-world conditions, and if not, is there a better model of weak UIC that provides fee simplicity without requiring honesty? If it is productive, how does this notion compare to weak incentive compatibility under $\gamma$-strict utility or other corrected utility metrics?
    \item What is the best update rule for the demand proxy (e.g. base fee) in a demand-modeling mechanism such as EIP-1559?
    \item What is the optimal tradeoff surrounding the failure regime of EIP-1559, and how do we mitigate deviations in those conditions? Alternatively, are there ways of minimizing that regime altogether? (this may just mean a better update rule)
    \item How can we tractably reason about mechanisms in their full iterated extent?
    \item Towards (5) in \ref{tfmprops}, can we design mechanisms that allow users more expressivity with respect to their preferences than just a linear sense of value related to being confirmed in this block or not? E.g. can we design a TFM that allows MEV-seeking players to directly pay for transaction ordering without being in the mining role?
    \item Given that burning is provably required for some incentive compatibility notions, are there robust ways of paying value forward to miners without encouraging fake transactions and OCAs?
\end{enumerate}

% SCRATCH michael

% What happens if we change the transaction fee mechanism itself? Explicitly, we could modify the above mining game in a few different ways without changing the axiomatic properties that \textit{(1) users pay minors for their work in keeping the ledger consistent, and (2) some users can pay more to get their transaction into the ledger sooner}. (we will later consider relaxing (2) to include payments for other nuanced user desiderata TODO) Potential changes include:
% \begin{enumerate}
%     \item Additional (for fairness, global) signals are posted before users submit transactions, which they can optionally condition their fee values on
%     \item Users submit a fee $f$ but if their transaction is included, they actually pay $f'$
%     \item Users submit a fee $f$ but if their transaction is included, the miner actually receives $f'$
%     \item TODO add more
% \end{enumerate}
% The idea is that, by relaxing the mechanism itself, perhaps we may be able to more robustly or elegant solve the issues presented in \ref{attacks}.

% What are those issues? A strong form of prevention against a large class of attacks, in which the above examples lie (TODO CHECK), would ensure the following properties:

\printbibliography

\end{document}